\documentclass[11pt]{article}
\usepackage{amsmath,amssymb,amsthm,cite}
\usepackage{fullpage}
\usepackage{cleveref}

\usepackage{color}
\usepackage{algorithm}
\usepackage[noend]{algpseudocode}
\usepackage{xspace}

\newtheorem{theorem}{Theorem}

\newtheorem{definition}{Definition} 
\newtheorem{lemma}[theorem]{Lemma}

\newcommand{\floor}[1]{\lfloor #1 \rfloor}

\newcommand{\clusterone}{\textsc{Cluster1}\xspace}
\newcommand{\clustertwo}{\textsc{Cluster2}\xspace}
\newcommand{\clusterthree}{\textsc{Cluster3}\xspace}

\newcommand{\FullOrShort}{full}
\ifthenelse{\equal{\FullOrShort}{full}}{
	
	  \newcommand{\fullOnly}[1]{#1}
	  \newcommand{\shortOnly}[1]{}

  }{
	  \newcommand{\fullOnly}[1]{}
	  \newcommand{\shortOnly}[1]{#1}
  }

\begin{document}

\title{Optimal Gossip with Direct Addressing\\[0.1cm] \normalsize {Regular Submission}}

\author{Bernhard Haeupler \footnote{Address: Microsoft Research, 1065 La Avenida, Mountain View, CA 94043, USA, Tl: +1-650-693-0214}\\Microsoft Research\\ \texttt{haeupler@alum.mit.edu} \and Dahlia Malkhi\footnote{Address: Microsoft Research, 1065 La Avenida, Mountain View, CA 94043, USA, Tl: +1-650-693-1362}\\Microsoft Research\\ \texttt{dalia@microsoft.com}}

\date{}

\maketitle

\shortOnly{
\thispagestyle{empty}
\setcounter{page}{0}
}

\begin{abstract}

Gossip algorithms spread information in distributed networks by having nodes repeatedly forward information to a few random contacts. By their very nature, gossip algorithms tend to be \emph{distributed} and \emph{fault tolerant}. If done right, they can also be \emph{fast} and \emph{message-efficient}. A common model for gossip communication is the random phone call model, in which in each synchronous round each node can PUSH or PULL information to or from a random other node. For example, Karp et al. [FOCS 2000] gave algorithms in this model that spread a message to all nodes in $\Theta(\log n)$ rounds while sending only $O(\log \log n)$ messages per node on average. They also showed that at least $\Theta(\log n)$ rounds are necessary in this model and that algorithms achieving this round-complexity need to send $\omega(1)$ messages per node on average.

\smallskip

Recently, Avin and Els\"asser [DISC 2013], studied the random phone call model with the natural and commonly used assumption of \emph{direct addressing}. Direct addressing allows nodes to directly contact nodes whose ID (e.g., IP address) was learned before. They show that in this setting, one can ``break the $\log n$ barrier'' and achieve a gossip algorithm running in $O(\sqrt{\log n})$ rounds, albeit while using $O(\sqrt{\log n})$ messages per node.

\smallskip
In this paper we study the same model and give a simple gossip algorithm which spreads a message in only $O(\log \log n)$ rounds. We furthermore prove a matching $\Omega(\log \log n)$ lower bound which shows that this running time is best possible. In particular we show that any gossip algorithm takes with high probability at least $0.99 \log \log n$ rounds to terminate. Lastly, our algorithm can be tweaked to send only $O(1)$ messages per node on average with only $O(\log n)$ bits per message. Our algorithm therefore simultaneously achieves the optimal round-, message-, and bit-complexity for this setting. As all prior gossip algorithms, our algorithm is also robust against failures. In particular, if in the beginning an oblivious adversary fails any $F$ nodes our algorithm still, with high probability, informs all but $o(F)$ surviving nodes.
\end{abstract}

\newpage

\section{Introduction}\label{sec:intro}

Gossip algorithms, also called rumor spreading protocols, spread information in a network by having nodes repeatedly forward (limited amounts of) information to a few, often randomly selected, contacts. They have found countless applications in solving coordination and information dissemination tasks in large distributed systems or networks. Their advantages in this setting over alternative approaches is manifold: They are naturally distributed and due to their randomized nature often very \emph{fault tolerant} with respect to many kinds of failures, such as, permanent or transient link-failures or node crashes. If done right, they can also be very efficient in their round-, message-, and bit-complexity.

\smallskip

Two models for gossip communication that have been considered and intensely studied in the literature are the \emph{random phone call model} and \emph{direct addressing}.

\smallskip

The \emph{random phone call model} assumes a complete network and synchronous rounds in which each node is able to PUSH or PULL information to or from a random other node. Both variants spread a message in $O(\log n)$ rounds with $O(\log n)$ messages sent per node \cite{pittel1987spreading} if every node communicates in every round. While the logarithmic running time is best possible for this setting, Karp et al. showed in an influential paper \cite{KarpSSV2000}, that the message-complexity could be reduced to $O(\log \log n)$ message per node on average. They also proved that any address-oblivious algorithm with logarithmic running time needs to send at least $\Omega(\log \log n)$ messages per node on average. Here \emph{address-oblivious} means that the decision of if and what a node sends is independent of the ID of its communication partner(s) this round (see also \cite{AvinElDISC13}). Lastly, they demonstrated the \emph{fault-tolerance} of their gossip algorithm. In particular, they remarked that any oblivious $F$ node failures still lead to all but $O(F)$ nodes being informed in the end, with high probability. Since then the random phone call model has been generalized in many directions, most notably to non-complete network topologies. For example, in \cite{GiakkoupisS2012vertex,Giakkoupis2011,NC} arbitrary topologies are assumed and the stopping time of the uniform PUSH-PULL gossip is studied where in each round each node initializes a communication with a uniformly random neighbor.

\smallskip

\emph{Direct addressing} is another natural assumption that has been intensely studied. Direct addressing assumes that nodes have \emph{unique IDs} (e.g., IP addresses) and that in each synchronous round a node can contact any node (or neighbor) whose ID (or IP address) is known to it. Since nodes do not know the topology and are only able to initiate one communication at a time, many algorithms for this setting let nodes contact a random node known to them. For example, in the famous Name-Dropper algorithm of Harchol-Balter et al. \cite{harchol1999resource}, nodes repeatedly forward all IDs they know to a random node known to them. They show that this simple algorithm terminates in $O(\log^2 n)$ rounds (using $O(\log^2 n)$ messages per node) if started from any (weakly) connected initial topology. Recently, \cite{kutten2003deterministic} gave an improved deterministic algorithm which achieves $O(\log n)$ rounds for the same setting. A similar setting but restricting a message to only contain one contact was considered in \cite{gossipdiscovery}. Lastly, direct addressing has also been extensively studied in general topologies. For example in \cite{CensorHillelShachnai2011,STOCversion,Haeupler2013} it is assumed that all possible communications are given by an arbitrary unknown topology in which nodes only know their neighbors.

\smallskip

Recently Avin and Els\"asser \cite{AvinElDISC13} showed that adding the natural direct addressing assumption in the random phone call model allows for gossip algorithms that ``break the $\log n$ barrier''. Their algorithm achieves a running time of $O(\sqrt{\log n})$ however at the cost of having a larger message-complexity of $\Theta(\sqrt{\log n})$ messages per node and also a larger bit-complexity of $O(n \log^{3/2} n + n b \log \log n)$ (compared to the $O(n b \log \log n)$ bit-complexity of \cite{KarpSSV2000}). Their algorithm is fault-tolerant in the same way as discussed for \cite{KarpSSV2000} above:

\smallskip

\begin{theorem}[Theorem 1 of \cite{AvinElDISC13}]\mbox{}\\
In the random phone call model with direct addressing there is a randomized fault-tolerant gossip algorithm that spreads a message of $b = \Omega(\log n)$ bits in $O(\sqrt{\log n})$ rounds using $O(\sqrt{\log n})$ messages per node and a total bit-complexity of $O(n \log^{3/2} n + n b \log \log n)$ bits  with high probability.
\end{theorem}

\subsection{Our Results}

In this paper we study the same setting and improve upon the state of the art by providing a drastically faster algorithm with simultaneously improved message- and bit-complexity:

\begin{theorem}\label{thm:main}\mbox{}\\
In the random phone call model with direct addressing there is a randomized fault-tolerant gossip algorithm that spreads a message of $b = \Omega(\log n)$ bits in $O(\log \log n)$ rounds using $O(1)$ messages per node on average and a total bit-complexity of $O(n b)$ bits with high probability.
\end{theorem}

In addition to the faster running time, it is easy to see that the $O(1)$ message-complexity per node and the total $O(nb)$ bit-complexity are optimal. In contrast to \cite{AvinElDISC13}, which performed worse than \cite{KarpSSV2000} on these complexity measures, our algorithm actually improves on the bounds of \cite{KarpSSV2000} and even breaks both message-complexity lower bounds given in \cite{KarpSSV2000} for the random phone call model without direct addressing. Lastly, we also prove a matching $\Omega(\log \log n)$ lower bound on the running time of any gossip algorithm in the random phone call model with direct addressing. This shows that our gossip algorithm achieves the optimal round-, message-, and bit-complexity simultaneously.

Our lower bound is particularly strong in the following sense: Instead of showing that any algorithm that takes $c \log \log n$ rounds for a sufficiently small constant $c$ fails with non-negligible probability we prove that any algorithm that takes $0.99 \log \log n$ rounds or less fails with high probability. Our lower bound furthermore holds irrespectively of many model assumptions and even applies to non-gossip algorithms which can contact more than one known node per round:

\begin{theorem}\label{thm:LB}
In the random phone call model with direct addressing any (randomized) gossip algorithm fails to spread a message to all nodes with high probability if less than $0.99 \log \log n$ rounds are used. This remains true if nodes are allowed to use arbitrarily large messages, communicate non address-obliviously and even contact an unbounded number of known nodes per round instead of just one. 
\end{theorem}

As a last contribution we consider bounds on the number of communications done by a node per round. While in a gossip algorithm each node initializes only one communication per round, nodes might have to respond to multiple requests. In fact, the number of simultaneous requests $\Delta$ can be as large as $n-1$. While this asymmetry is crucially used in essentially all algorithms that use direct addressing, it is unfortunately a questionable assumption for at least some systems and applications (see \Cref{sec:boundedcommunication} for a more detailed discussion). In this regard, we discuss the easy to prove fact that a sub-logarithmic running time inherently requires $\Delta = \omega(1)$ and show that the best running time possible for any given $\Delta$ is $\log n / \log \Delta$. We then give a general framework for working efficiently in a setting with bounded $\Delta$. In particular, we show how our gossip algorithm can efficiently compute a $\Delta$-clustering which can then be used to achieve any point on this trade-off curve between running time and round complexity (subject only to our $\Omega(\log \log n)$ lower bound from \Cref{thm:LB}). This shows, for example, that one can still have  gossip algorithms with an optimal $\Theta(\log \log n)$ running-time while keeping $\Delta$ smaller than any polynomial in $n$, that is, $\Delta = n^{o(1)}$. We believe this to be a useful setting of parameters which seems realistic in many settings.

\begin{theorem}\label{thm:deltaclustering}
For any $\Delta = \log^{\omega(1)} n$ there is a randomized gossip algorithm that computes a $\Delta$-clustering in $O(\log \log n)$ rounds using $O(1)$-messages per node on average but having no node respond to more than $\Delta$ requests per round. Given such a $\Delta$-clustering a message broadcast can then be performed easily with an optimal round-, message-, and bit-complexity.
\end{theorem}

\paragraph{Organization}

The rest of the paper is organized as follows: 
 In \Cref{sec:clustering} we introduce \emph{clusterings}, which are the main abstraction used to design our gossip algorithms. In \Cref{sec:simplealg} we give the simplest version of our algorithm which demonstrates the main ideas needed to achieve the $\Theta(\log \log n)$ round complexity. In \Cref{sec:linear} we then explain how to reduce the number of messages sent to $O(1)$ per node on average and achieve an optimal bit-complexity which proves our main theorem. In \Cref{sec:LB} we prove \Cref{thm:LB}, i.e., our $\Omega(\log \log n)$ round lower bound. In \Cref{sec:boundedcommunication} we provide our results on gossip algorithms with a bounded amount of communications per node. Lastly, in \Cref{sec:faults} we discuss the fault-tolerance of our algorithms.

\section{The Random Phone Call Model and Direct Addressing}\label{sec:model}

In this section we give the formal definitions for the random phone call model and the direct addressing assumption used throughout this paper.

\paragraph{The Network}

The random phone call model assumes a complete network consisting of $n$ nodes. Each node has a unique \emph{address} or \emph{ID} consisting of $O(\log n)$ bits, that is, we assume a polynomially large ID space. Initially nodes know their own ID but not the ID of any other node in the system. For convenience we assume that nodes know the number of nodes $n$. This assumption is without loss of generality since for all problems considered in this paper it is easy to test with high probability whether the algorithm succeeded. This allows for determining the parameter $n$ using the classical guess-test-and-double strategy without increasing the running times by more than a constant factor.

\paragraph{Communication}

Communication takes place in synchronous rounds. In each round a node can choose to initialize one communication. In particular, it can either choose to contact a random node or use direct addressing to contact a node whose ID it knows. Each such contact from a node $u$ to a node $v$ consists of either node $u$ PUSHing a message to $v$ or of node $u$ PULLing a message from $v$. A message can in principle contain any number of bits but all\footnote{The only exceptions are when we share a $b$-bit message (via a ClusterShare) and for resizing clusters in the Algorithm \clusterone, which demonstrates the ideas for achieving the optimal $\Theta(\log \log n)$ round complexity while ignoring message counts and message sizes for simplicity.} algorithms throughout this paper only use messages of (minimal) size containing $O(\log n)$ bits. In particular, every message carries either the information to be broadcast, a node count, or $O(1)$ node IDs. Following \cite{AvinElDISC13} (and \cite{KarpSSV2000}) we say a gossip algorithm is \emph{address-oblivious} if the decision whether to send or not and if so what to send (or respond with) is independent of the communication partner(s) in this round. In particular this means that for every node $u$ the messages PUSHed by $u$ does not depend on the ID of the communication partner and all responses $u$ gives to a PULL are the same in one round. All algorithms presented in this paper are address-oblivious. 

\paragraph{Round-, Message-, and Bit-Complexity}

The important figures of efficiency for such a gossip algorithm are the round-complexity, that is, the number of synchronous rounds needed, the message-complexity, that is, the number of messages sent per node on average throughout the execution, and the bit-complexity, that is, the total number of bits communicated in all messages throughput the execution. We only consider and present randomized protocols which complete their task with high probability in the round-, bit-, and message-complexity stated, that is, with probability at least $1 - \frac{1}{n^{C}}$ for some large chosen constant $C > 1$.

\paragraph{The (Fault-Tolerant) Broadcast Task}

Gossip algorithms can be used for a multitude of different communication and coordination tasks. The algorithms in this paper compute a cluster containing all nodes or, in the case of \Cref{thm:deltaclustering}, a $\Delta$-clustering which can then be used to perform any of these tasks easily and efficiently. To be consistent with prior work we focus on one specific task, namely, broadcasting. In this task there is a message, often also called rumor, consisting of $b$ bits with $b = \Omega(\log n)$ which is initially known to one node (or multiple nodes). The goal is to inform all nodes in the network. In \Cref{sec:faults} we also address the setting of node-failures. In this setting an oblivious adversary fails any $F$ nodes at the beginning of the execution. Follow \cite{KarpSSV2000} and \cite{AvinElDISC13}, the desired guarantee in this setting is that all but $o(F)$ surviving nodes are clustered/informed. We call an algorithm achieving this \emph{fault-tolerant}.

%
%

\section{Clusterings}\label{sec:clustering}

In this section we introduce clusterings, the main tool used in this paper to achieve efficient gossip.

First, in \Cref{sec:clusterdef}, we define the concept and implementation of a clustering. In \Cref{sec:primitives} we then explain how these clustering support efficient implementations of several primitives allowing for coordination and communication between clusters. These primitives are extensively used throughout this paper to implement our fast and message efficient gossip communication algorithms.

\subsection{Definition and Implementation of Clusterings}\label{sec:clusterdef}

A \emph{clustering} is a partition of all nodes $v \in V$ in a network into disjoint \emph{clusters} $C_1,C_2,\ldots,C_k$ and the set $U = V \setminus \bigcup_i C_i$ of \emph{unclustered nodes}. Each cluster $C$ has one \emph{cluster leader} $l_C \in C$ known to all nodes in $C$. We call all clustered nodes that are not leaders \emph{cluster followers}. We define the \emph{ID of a cluster} to be the ID of its leader and the \emph{size of a cluster} $C$ to be the number of nodes $|C|$ it contains.

Clusterings are implemented by each node $v \in V$ containing a variable $follow_v$. For any clustered node $v \in C$ this variable contains the ID of the cluster leader $ID(l_c)$ while for all unclustered nodes $v' \in U$ it is set to $\infty$. This also allows nodes to determine whether they are a cluster leader of follower by simply comparing their own ID to their $follow$ variable.

\subsection{Cluster Coordination Primitives}\label{sec:primitives}

Having a central leader which is known to all nodes in a cluster allows a cluster to coordinate using only a constant number of communication rounds. In one communication round, followers may PUSH information to the leader and thus have the leader collect inputs from the entire cluster. In a second round, followers may PULL a response from the leader. Below, we omit the details concerning PUSHing and PULLing messages, and briefly refer to exchanges such as above as followers sending information (if any) to the leader and pulling one directive from it.

Following is a list of cluster-macros that make use of the power of clusters, and will be used later in our constructions. All these cluster primitives require only a constant number of rounds and therefore also a constant number of messages per node. 

\begin{itemize}

\itemsep0em

	\item ClusterActivate($p$)\\	 This primitive allows clusters to be activated independently with probability $p$ by having followers pull the outcome of a $p$-biased coin flipped by their leader. Each cluster stays activated until the next call of ClusterActivate or until it is explicitly deactivated.
	
	\item ClusterSize\\ Each cluster can determine its size in two rounds by each cluster follower sending its ID to the cluster leader, the cluster leader counting the received IDs and all followers pulling this count from the leader.

	\item ClusterDissolve($s$)\\ This primitive ensures that all clusters are at least of the given size $s$ and runs in two rounds: In the first round each cluster follower pushes their ID to the leader. The cluster leader uses these messages to determine the cluster size $s'$. In the second round followers pull a value indicating their new clustering status from their current leader and set their $follow$ variable to this value. If $s' \geq s$ the cluster leader responds with its own ID and otherwise responds with $\infty$ to dissolve the cluster. It also sets its own $follow$ variable the same way. 

	\item ClusterResize($s$)\\ This primitive ensures that all clusters are of size at most $2s$. It runs in two rounds. In the first round each cluster follower pushes their ID to the leader. This allows the leader of any cluster to determine the size $s'$ of its cluster. In the second round followers pull a message from their current leader and update their $follow$ variable as follows: The cluster leader (in his mind) re-clusters its followers into $\floor{\frac{s'}{s}}$ clusters of equal size (up to one) declaring for each new cluster the node with the largest ID its leader. To announce\footnote{The messages used in this primitive are of size $O(\frac{s'}{s} \log n)$ bits. This calls for ensuring that $\frac{s'}{s} = \Theta(1)$ if a small bit complexity is desired. We note that this complication arises only for address-oblivious algorithms as a non address-oblivious algorithm could simply respond to each follower separately with the ID of its new leader.} and establish this clustering the leader responds to each node's pull request with the IDs of the $\floor{\frac{s'}{s}}$ new leaders. Each follower then sets its $follow$ variable to the largest leader ID that is at least as large as its own ID. It is easy to see that after a cluster resizing step all clusters have size at most $2s -1$.	
	
\item ClusterPUSH($msg$)/ClusterPULL($msg$)\\ This is a similar primitive to a random PUSH or PULL of a node except on a cluster level. Each cluster leader decides on whether a push or pull is to be performed and what message $msg$ to send or with what message $msg$ to respond if pulled. This decision is then shared within each cluster by followers pulling it from the leader. All nodes in clusters instructed to do now contact a random node and either relay the message of the cluster leader via a PUSH or request a message via a PULL. All messages received by a node through either a PULL or PUSH get then relayed to their cluster leader (if there is one).

\item ClusterMerge(new leader ID)\\ This primitive allows to merge clusters in one round. For this all followers simply request the ID of the new leader from their current leader to update their $follow$ variable. In addition to answering these pulls the leader also sets its own $follow$ variable to the new leader.

	\item ClusterShare(.)\\	 This primitive allows clusters to share information within a cluster in two rounds. In the first round, followers who have relevant information send this information to the leader, who aggregates the information. In a second round, followers pull the aggregated information from the leader.

\end{itemize}

\section{An $O(\log\log n)$ Round Gossip Algorithm}\label{sec:simplealg}

In this section we present the simplest version of our algorithm to demonstrate the principal ideas needed achieve the optimal $O(\log \log n)$ round-complexity. The resulting algorithm \clusterone is, for sake of simplicity, not tuned to have a good message- or bit-complexity. In particular, a constant fraction of the nodes is sending during most of the $\Theta(\log \log n)$ rounds and the message sizes occurring in the first ClusterResize call is hard to analyze.

\subsection{Overview}\label{sec:clusteroneOverview}

The idea of \clusterone~is simple, and a general outline can be inferred from the procedure calls and comments at the top of \Cref{alg:cluster1}: To obtain an initial clustering, Procedure {\sc GrowInitialClusters} samples a small logarithmic fraction of the nodes to active singleton-clusters. Then a standard PUSH gossip is performed for $\Theta(\log \log n)$ rounds in which these clusters recruit unclustered nodes. The usual exponential growth behavior of the PUSH gossip protocol leads to clusters of size at least $s = C' \log n$ while also ensuring that all but a small constant fraction of the nodes are recruited into some cluster. In the main phase of the algorithm, namely Procedure {\sc SquareClusters}, we repeatedly square this cluster size by merging clusters in the following manner: First, we use ClusterDissolve($s$) and ClusterResize($s$) to ensure that clusters have sizes between $s$ and $2s$. Then we determine the new cluster leaders by activating each cluster with probability $1/s$. The inactive clusters are then prompted to join one of the new active clusters via two ClusterPUSHes. This results in each new cluster recruiting $\Theta(s)$ inactive clusters which leads to the desired new cluster size of $\Theta(s^2)$. In particular, the first such PUSH guarantees that each new cluster grows by a factor of $\Theta(s)$, while the second iteration ensures that all old inactive clusters are merged into some new, larger active cluster. The doubly exponential growth of squaring the cluster size with every iteration guarantees that $\Theta(\log \log n)$ iterations are sufficient to obtain clusters of size $\frac{\sqrt{n}}{\log n}$. Procedure {\sc MergeAllClusters} uses two ClusterPUSHes to merge all clusters into the cluster with the smallest ID. With one cluster comprising most of the nodes in the network, Procedure {\sc UnclusteredNodesPull} uses a standard PULL gossip to have the remaining unclustered nodes join the cluster which is known to require only $\Theta(\log \log n)$ rounds. Lastly, with all nodes in one cluster spreading any desired message can be done via one ClusterShare.

A breakdown of the various steps is given in the detailed procedures of \Cref{alg:cluster1}, using cluster communication primitives from \Cref{sec:primitives}. 

\begin{algorithm}[htb!]
\caption{\clusterone}
\begin{algorithmic}[1]
\State \Call{GrowInitialClusters}{} \Comment{cluster 90\% of all nodes into clusters of size at least $C'\log n$}
\State \Call{SquareClusters}{} \Comment{repeatedly merge clusters to get to size $\frac{\sqrt{n}}{\log n}$}
\State \Call{MergeAllClusters}{} \Comment{merge all clusters to one}
\State \Call{UnclusteredNodesPull}{} \Comment{unclustered nodes pull to join the cluster}
\State ClusterShare(message)

\Statex
\Procedure{GrowInitialClusters}{}
\State With probability $1 / C \log n$ set $follow$ to $ownID$ otherwise to $\infty$
\For{$\Theta(\log\log n)$ times}
	\State ClusterPUSH(clusterID)
	\State unclustered nodes: set $follow$ to any received ID (if any)
  %
\EndFor{}
\EndProcedure

\Statex 

\Procedure{SquareClusters}{}
\State $s = C'\log n$
\State ClusterDissolve($s$)
\Repeat 
  \State ClusterResize($s$)   \label{cl1:clusterresize}
  \State ClusterActivate($1/s$)
  \For{two repetitions} \Comment{active clusters recruit all inactive clusters}
    \State active clusters: ClusterPUSH($follow$)
    \State inactive clusters: ClusterMerge(smallest received ID (if any))
  \EndFor{}
 $s \gets \Theta(s^2)$
\Until{$s > \frac{\sqrt{n}}{\log n}$}
\EndProcedure
	
\Statex
\Procedure{MergeAllClusters}{} 
  \For{two repetitions} \Comment{clusters with smallest ID recruits all other clusters}
    \State ClusterPUSH($follow$)
    \State ClusterMerge(smallest received ID)
  \EndFor{}
\EndProcedure

\Statex
\Procedure{UnclusteredNodesPull}{} 
\State repeat $\Theta(\log\log n)$ times: unclustered nodes PULL $follow$ value
\EndProcedure

\end{algorithmic}
\label{alg:cluster1}
\end{algorithm}

\subsection{Correctness Proof}

In this section we prove the correctness and efficiency of \clusterone, which is described in \Cref{sec:clusteroneOverview} and \Cref{alg:cluster1}. 
We break the proof down to several lemmas, which then lead to the main result. \shortOnly{Due to space constraints, we defer proof details to \Cref{app:proofscluster1}; the statements of the lemmas give the gist of the proof rationale.}



\begin{lemma}\label{lem:initialclusters}
At the end of Procedure {\sc GrowInitialClusters} in \clusterone, at least $0.9n$ nodes are contained in clusters of size $C\rq{} \log n$ with high probability.
\end{lemma}

\newcommand{\proofinitialclusters}{
\begin{proof}\shortOnly{[Proof of \Cref{lem:initialclusters}]}
The clustering at the end of the procedure is created by $\Theta(\log \log n)$ rounds of random PUSH gossip starting from cluster centers that were obtained by sampling nodes independently with probability $p = \frac{1}{C\log n}$. By a Chernoff bound, we initially have the expected $\Theta\left(\frac{n}{C\log n}\right)$ cluster centers occur with high probability. By standard gossip analysis, the set containing all clustered nodes grows by a constant factor in each round, so long as some constant fraction of the system remains unclustered. Furthermore, by choosing $C' \ll C$ at most an arbitrary small constant fraction of nodes account for clusters that are smaller than $C' \log n$, namely about $\frac{n}{C\log n} \cdot C' \log n$ many. In total, one can guarantee that at least a $90\%$ fraction of nodes are good, that is, contained in a large cluster. 
\end{proof}
}\fullOnly{\proofinitialclusters}

\begin{lemma}\label{lem:clusteroneprog}
Every iteration inside Procedure {\sc SquareClusters} in \clusterone squares the cluster size $s$ while ensuring that all clustered nodes are contained in clusters of the new size, with high probability.
\end{lemma}

\newcommand{\proofclusteroneprog}{
\begin{proof}\shortOnly{[Proof of \Cref{lem:clusteroneprog}]}
Initially, due to \Cref{lem:initialclusters} which is followed by ClusterDissolve($s$), we can assume that at least $0.9n$ nodes are clustered and divided into clusters of size between $s$ and two $2s$ after the ClusterResize($s$) in Line~\ref{cl1:clusterresize} is performed. The same holds from there on by induction. These $\Theta(n/s)$ clusters then get independently activated with probability $1/s$. Since $s \leq \frac{\sqrt{n}}{\log n}$ the expected number of activated clusters is $\Theta(n/s^2) = \omega(\log n)$ and a Chernoff bound shows that indeed $\Theta(n/s^2)$ clusters are activated with high probability. Now, in the first ClusterPUSH/ClusterMerge iteration of Procedure {\sc SquareClusters}, these active clusters PUSH $\Theta(s)$ messages each for a total of $\Theta(n/s)$ messages in order to recruit the $\Theta(n/s)$ inactive clusters. By standard balls and bins analysis, a linear fraction of every set of $s = \Omega(\log n)$ PUSHes of each cluster lands in a unique cluster and result in a merge. Therefore, after the first iteration all active clusters will recruit $\Theta(s)$ inactive clusters and thus grow to size $\Theta(s^2)$. However, there is also a constant fraction of inactive clusters that did not get recruited. The second ClusterPUSH/ClusterMerge iteration ensures that these clusters merge with some active cluster to ensure that all clustered nodes remain clustered. This is indeed the case since there are $\Theta(n)$ total number of nodes contained in the active clusters which PUSH out in the second iteration. Every inactive cluster of size $s= \Omega(\log n)$ is therefore reached and merged into an active cluster with high probability.
\end{proof}
}\fullOnly{\proofclusteroneprog}

\begin{lemma}\label{lem:mergeallclusters}
After Procedure {\sc MergeAllClusters} in \clusterone, with high probability, there is a single cluster which contains at least $0.9 n$ nodes.
\end{lemma}

\newcommand{\proofmergeallclusters}{
\begin{proof}\shortOnly{[Proof of \Cref{lem:mergeallclusters}]}
\Cref{lem:clusteroneprog} guarantees that before Procedure {\sc MergeAllClusters} is invoked at least $0.9 n$ nodes are contained in clusters of size $\Theta(s)$ with $s \geq \frac{\sqrt{n}}{\log n}$. We now consider the cluster with the smallest ID. The first ClusterPUSH/ClusterMerge iteration of Procedure {\sc MergeAllClusters} will, with high probability, grow this cluster to size $\Omega(\frac{n}{\log^2 n})$. The second iteration then recruits all remaining clusters, with high probability.
\end{proof}
}\fullOnly{\proofmergeallclusters}

Lastly, it is well known (see, e.g., \cite{KarpSSV2000}) that the standard PULL gossip takes only $O(\log \log n)$ rounds to inform the remaining constant or even $1 - 1/\log^{\Theta(1)} n$ fraction of nodes. This is exactly what is used in Procedure {\sc UnclusteredNodesPull} to have all unclustered nodes join the single cluster in the end:

\begin{lemma}\label{lem:pull}
The PULL gossip in Procedure {\sc UnclusteredNodesPull} of \clusterone has all nodes join the single cluster, with high probability, in $O(\log \log n)$ rounds.
\end{lemma}

\newcommand{\proofpull}{
\begin{proof}\shortOnly{[Proof of \Cref{lem:pull}]}
It is well known that the standard uniform PULL gossip terminates with high probability in $\Theta(\log \log n)$ if one starts with at least $n / \log^{\Theta(1)} n$ nodes being informed. For completeness we include the analysis here:
We look at the fraction $x$ of unclustered nodes and note that in one round any unclustered node (pulling a random node) remains unclustered with probability $x$. As long as the number of unclustered nodes is still at least $\Theta(\log n)$ a Chernoff bound shows that the quantity $x$ shrinks to at most $2x^2$ with high probability after each pull. This can happen at most for $2\log \log n$ rounds. Lastly, with high probability, it takes at most a constant number of rounds until each remaining node pulls from a clustered node. This shows that the $\Theta(\log \log n)$ rounds of PULL gossip indeed, with high probability, lead to one cluster containing all nodes. 
\end{proof}
}\fullOnly{\proofpull}

\noindent The correctness and efficiency of \clusterone now follows immediately:

\begin{theorem}\label{thm:clusterone}
\clusterone terminates after $O(\log \log n)$ rounds and with high probability spreads any message to all nodes.
\end{theorem}

\newcommand{\proofclusterone}{
\begin{proof}\shortOnly{[Proof of \Cref{thm:clusterone}]}
The ClusterShare(message) performed at the end of \clusterone correctly spreads any message to all nodes since according to \Cref{lem:pull} all nodes are contained in one single cluster. The running time of \clusterone is furthermore easily seen to be $O(\log \log n)$ rounds. In particular, the loops inside Procedures {\sc GrowInitialClusters} and {\sc UnclusteredNodesPull} run explicitly for $O(\log \log n)$ iterations, while the loop inside Procedure {\sc SquareClusters} squares the parameter $s$ in each iteration, and therefore runs for at most $\Theta(\log \log n)$ iterations as well. To complete the proof, we refer to \Cref{sec:primitives} for the fact that all cluster primitives performed in these loops can be performed in a constant number of rounds.
\end{proof}
}\fullOnly{\proofclusterone}

\section{Achieving Optimal Linear Message Complexity}\label{sec:linear}

In this section we prove \Cref{thm:main} by explaining how the \clusterone algorithm can be tweaked to also have optimal message- and bit-complexity without compromising its optimal $O(\log \log n)$ round-complexity. We call the resulting algorithm \clustertwo.

\subsection{Overview}\label{sec:clustertwoOverview}

First, we give the intuition behind achieving a linear message complexity and describe in which aspects \clustertwo is different from \clusterone. The high-level recipe of \clustertwo is the same as \clusterone. The key idea for achieving the optimal message-complexity is in modifying Procedures {\sc GrowInitialClusters} and {\sc SquareClusters} to work with only $\Theta(n/ \log n)$ clustered nodes. This guarantees that even if all clustered nodes actively transmit in these rounds, only $o(n)$ messages are used for these parts. After merging all clusters, we are left with one large cluster of size $\Theta(n/\log n)$. While the standard PULL gossip would still finish in $O(\log \log n)$ rounds this would again take to many messages (of unclustered nodes PULLing unclustered nodes). We therefore PUSH the cluster ID to $\Theta(n)$ nodes first before doing the final PULL phase. With so many informed nodes each node expects to PULL at most a constant number of times and it is easy to see that in total at most $O(n)$ messages are used with high probability. 

All in all a much more delicate control on the number of clusters, their size and the number of clustered nodes is required. This is also important to provably guarantee an optimal bit-complexity. The subtle problem here is that a ClusterResize($s$) on a cluster of size $\omega(s)$ leads to messages of size $\omega(\log n)$ (see footnote in \Cref{sec:primitives}). This control is achieved as follows:

To control the number of nodes clustered in the Procedures {\sc GrowInitialClusters} we have clusters measure their growth. More precisely before and after each PUSH in the initial recruiting phase each cluster checks its size. If a cluster has size at least $\Omega(\log^3 n)$ but grew by less than a factor of $2 - \frac{1}{\log n}$ then it stops recruiting. A Chernoff Bound shows that no cluster will stop recruiting if less than  $n / 2\log n$ nodes are clustered.  On the other hand, similar arguments as for \clusterone show that indeed $\Theta(n / \log n)$ nodes will get recruited of which most lie in clusters of size $s = \Omega(\log^3 n)$. A ClusterResize($C' \log n)$ called for every large cluster in every iteration furthermore ensures that no individual cluster gets too big. For the Procedure{\sc SquareClusters} we then again merge clusters. Due to the smaller number of clustered nodes most cluster requests now go to unclustered nodes and get ignored. Still, any iteration starting with clusters of size $s$ ends up with larger clusters of size $\Theta(s^2 / \log n) = \omega(s^{1.5})$. Growing the cluster sizes from $s = \Theta(\log^3 n)$ to $\Omega(n / \log^2 n)$ therefore still takes only $O(\log \log n)$ rounds. Chernoff bounds over the larger $\Omega(\log^3 n)$ sized clusters still guarantee that the new cluster size can be predicted up to a constant with high probability which ensures that no ClusterResize($s$) is applied to a too large (or too small) cluster. 

The formal description of \clustertwo is given as \Cref{alg:cluster2}\shortOnly{in \Cref{app:alg2}}.

\newcommand{\algorithmtwo}{
\begin{algorithm}[htb!]
\caption{\clustertwo}
\begin{algorithmic}[1]

\State \Call{GrowInitialClusters}{} \Comment{cluster $\Theta(n/\log n)$ nodes into clusters of size $(1 + \Theta(1))(C'\log n)$} 
\State \Call{SquareClusters}{} \Comment{repeatedly merge clusters to get to size $\frac{\sqrt{n}}{\log^2 n}$}
\State \Call{MergeAllClusters}{} \Comment{as in \Cref{alg:cluster1}}
\State \Call{BoundedClusterPush}{} \Comment{clusters push to all but $\Theta(n)$ unclustered nodes}
\State \Call{UnclusteredNodesPull}{} \Comment{as in \Cref{alg:cluster1}}
\State ClusterShare(message)
	
\Statex
\Procedure{GrowInitialClusters}{}
\State With probability $1 / C \log^4 n$ set $follow$ to $ownID$ otherwise to $\infty$
\State ClusterActivate(1)
\For{$\Theta(\log\log n)$ times}  
	\State active clusters: ClusterPUSH($follow$)
	\State unclustered nodes: set $follow$ to any received ID (if any)
	\If {ClusterSize $ \geq C'\log^3 n$}
		\If {cluster grew by less than a factor of $(2 - 1/\log n)$}
			\State deactivate cluster
		\Else
			\State ClusterResize($C' \log^3 n$)
		\EndIf
	\EndIf
\EndFor{}
\EndProcedure

\Statex
\Procedure{squareClusters}{}
\State $s = C'\log^3 n$
\State ClusterDissolve($s$)
\Repeat 
  \State ClusterResize($s$) 
  \State ClusterActivate($1/s$)
  \For{two repetitions} \Comment{Active Clusters recruit all Inactive Clusters}
    \State active clusters: ClusterPUSH($follow$)
    \State inactive clusters: ClusterMerge(random received ID (if any))
  \EndFor{}
 $s \gets \Theta(s^2/\log n)$
\Until{$s > \frac{\sqrt{n}}{\log^2 n}$}
\EndProcedure
	
\Statex

\Procedure{BoundedClusterPush}{} 
\State ClusterActivate(1)
\For {$\Theta(\log \log n)$ rounds}
                \State ClusterPUSH($follow$)
                \State unclustered nodes: set $follow$ to any received ID (if any)
                \State ClusterSize
                \If{cluster grows by less than a factor of $1.1$}
									\State deactivate cluster
                \EndIf
\EndFor 
\EndProcedure

\end{algorithmic}
\label{alg:cluster2}
\end{algorithm}
}\fullOnly{\algorithmtwo}

\subsection{Correctness Proof}

In this section we prove \Cref{thm:main} by showing that \clustertwo spreads a message with the desired optimal round-, message- and bit-complexity. 

The proof follows along the same line as \Cref{thm:clusterone} while including the extra arguments given in \Cref{sec:clustertwoOverview}. We again break the proof into several lemmas, which then lead to the main result.

\begin{lemma} \label{lem:cl2samp}
If at most $n/f(n)$ nodes are clustered, then a ClusterPUSH by some cluster $X$, where $|X| \geq \Omega(\log n)$, grows $X$ by factor at least $(2-\Theta(1/f(n)))$ with high probability.
\end{lemma}

\newcommand{\proofcltwosamp}{
\begin{proof}\shortOnly{[Proof of \Cref{lem:cl2samp}]}
Each PUSH request of one of the $|X|$ nodes is hits a unique unclustered node with probability of approximately $(1 - \frac{1}{f(n)})$. The expected number of nodes recruited for cluster $X$ is therefore $(1 - \frac{1}{2f(n)})|X|$ for a total number of $|X| + (1 - \frac{1}{2f(n)})|X|$ and a  growth factor of $(2-\Theta(1/2f(n)))$.
Since the recruitment success of different nodes are essentially independent of each other applying a Chernoff type bound shows that the growth factor is close to the expected and in particular is $(2-\Theta(1/f(n)))$, with high probability.
\end{proof}
}\fullOnly{\proofcltwosamp}

\begin{lemma} \label{lem:cl2init}
Procedure {\sc GrowInitialClusters} of \clustertwo leads to $\Theta(n/\log n)$ clustered nodes that are contained in clusters of size $(1 + \Theta(1))(C'\log n)$, with high probability, while at most $O(n)$ messages of $\Theta(\log n)$ bits each are sent.
\end{lemma}

\newcommand{\proofcltwoinit}{
\begin{proof}\shortOnly{[Proof of \Cref{lem:cl2init}]}
For the initial sampling of singleton clusters, which happens with probability $1/C \log^4 n$, a Chernoff bound shows that $\Theta(\frac{n}{C\log^4 n})$ such singleton clusters are created, with high probability. We would like to argue as in \Cref{lem:initialclusters} that the growth of the set of clustered nodes is exponential until at least $\Omega(n/\log n)$ nodes are clustered. We need to show that the newly introduced size-control does not prevent this exponential growth. By \Cref{lem:cl2samp}, as long as at least $n (1 - O(1/\log n))$ nodes are unclustered each cluster will grow by a factor of at least $2 - 1/\log n$, with high probability. Hence, the size control mechanism does not inhibit the exponential growth until at least $\Theta(n/\log n)$ nodes are clustered. Since at most about $\frac{n}{C\log^4 n} \cdot (C' \log^3 n)$ nodes account for clusters that are smaller than $C' \log^3 n$, indeed at least $\Theta(n/\log n)$ nodes are contained in large clusters.

To guarantee that the message complexity incurred by the initialization is small we note that only clustered nodes are transmitting. This leads to at most $O(n \log \log n / \log n) = O(n)$ messages sent. Since every cluster can grow by at most a factor of two in each iteration it is furthermore clear that all message sizes used in a ClusterResize contain at most two IDs. This results in messages of size $\Theta(\log n)$, as claimed. 
\end{proof}
}\fullOnly{\proofcltwoinit}

\begin{lemma} \label{lem:cl2main}
Before and after every iteration inside Procedure {\sc SquareClusters} in \clustertwo, $\Theta(n/\log n)$ nodes are contained in clusters of size at least $s$, with high probability. Furthermore, the total number of messages sent during Procedure {\sc SquareClusters} is $O(n)$, each containing $\Theta(\log n)$ bits. 
\end{lemma}

\newcommand{\proofcltwomain}{
\begin{proof}\shortOnly{[Proof of \Cref{lem:cl2main}]}
We prove this by induction. The base case is already proved in \Cref{lem:cl2init}. Next, we show that each iteration of the loop increases the cluster size from $s$ to $s' = \Theta(\frac{s^2}{\log n})$ while keeping all clustered nodes clustered. In the first iteration each active cluster PUSHes to $\Theta(s)$ random nodes and hits with high probability $\Theta(s/\log n)$ inactive clusters (alone). Once these clusters merge, any active cluster therefore grows to $s' = \Theta(\frac{s^2}{\log n})$ as desired. The second PUSH, performed by the larger clusters which cover $\Theta(n/\log n)$ nodes, guarantees that no inactive cluster of size $s$ is left behind without merging to a new larger cluster of size at least $s'$. Furthermore, since inactive clusters randomly join the large active clusters during this second iteration no active cluster grows by more than a constant factor. The new clusters are therefore all of the desired size $s' = \Theta(\frac{s^2}{\log n})$ while
all clustered nodes remain clustered. Lastly, since $s > \log^3 n$ we have that $s' = \omega(s^{1.5})$ which shows that the number of iterations until the final cluster size of $\Omega(\frac{\sqrt{n}}{\log^2 n})$ is reached is at most $\Theta(\log \log n)$.

All in all only $\Theta(n / \log n)$ nodes are communicating in each of the $O(\log \log n)$ rounds which leads to the desired $O(n)$ messages-complexity. For the bit-size of these messages, the only thing to worry about is the calls to ClusterResize($s$). However, since the cluster sizes are, up to constants, exactly what they should be the ClusterResize($s$) produces only $\Theta(\log n)$ bit messages containing $\Theta(1)$ IDs.
\end{proof}
}\fullOnly{\proofcltwomain}

\noindent
\Cref{lem:mergeallclusters} applies without change to \clustertwo and we do not repeat it.

\begin{lemma}\label{lem:boundedpush}
Procedure {\sc BoundedPushCluster} expands the single large cluster to $\Theta(n)$ nodes in $O(\log\log n)$ rounds, with high probability, while at most $O(n)$ messages are sent.
\end{lemma}

\newcommand{\proofboundedpush}{
\begin{proof}\shortOnly{[Proof of \Cref{lem:boundedpush}]}
At the beginning of the procedure, we have from \Cref{lem:mergeallclusters} that $\Theta(n/\log n)$ belong to a single giant cluster. From here, we use a standard analysis of PUSH-PULL gossip: So long as at least $.1n$ nodes are unclustered, during each iteration the number of clustered nodes grows exponentially (and therefore for at most $O(\log \log n)$ rounds) until a large constant fraction of the nodes is clustered. By \Cref{lem:cl2samp}, the cluster detects this event successfully. Since the number of messages sent out in every round is exponentially increasing until it stops the total number of messages form a geometric sum which sums to $O(n)$ messages sent in total. 
\end{proof}
}\fullOnly{\proofboundedpush}

Putting these claims together, it follows that \clustertwo has a bit-complexity of $O(n \log n + nb) = O(nb)$ and message complexity of $O(1)$ messages per node on average while maintaining the optimal $O(\log \log n)$ round complexity of \clusterone. This proves \Cref{thm:main}.

\section{The $\Omega(\log \log n)$ Round Lower Bound}\label{sec:LB}

This section is dedicated to prove the $\Omega(\log \log n)$ round-complexity lower bound of \Cref{thm:LB}.

We first remark that the simple lower bound provided in \Cref{sec:boundedcommunication} as \Cref{lem:trivialLB} does not give any super constant number of rounds for the most interesting setting of $\Delta = n^{\Theta(1)}$. Here we show that the round-complexity of broadcasting a message has to be at least $0.99\log \log n$ with high probability. This holds for any $\Delta$ and remains true even if one ignores many restrictions of the model regarding message sizes, address-obliviousness, or the number of known nodes that can be contacted per round. The main implication of this lower bound is that both \clusterone and \clustertwo achieve the optimal asymptotic running time.

We need some notation: Instead of each node sampling (up to) one random contact in each round we assume that these nodes are sampled in advance. In particular, let $u_{v,t}$ be the uniformly random node given to node $v$ if it samples an node at time $t$. Based on these samples we define $G_t$ to be the union of all edges (node-pairs) that are potentially sampled at time $t$, that is, $G_t=(V,E_t)$ with $E_t = \{ \{u,v\} \ |  \ \exists v: \ u = u_{v,t} \}$.

\begin{lemma}\label{lem:knowledgegraph}
For every $t = 0,1,\ldots$ let $K_t$ be the graph indicating which node knows whom at the beginning of round $t$. It holds that:
$$K_0 = \emptyset,$$
$$K_{t+1} \subseteq (K_{t} \cup G_{t+1})^2,$$
and therefore also
$$K_{t} \subseteq \left(\bigcup_{i=1}^t G_i\right)^{{2}^{t}}.$$
Here $G^j$ is the graph that connects a node $u$ to any node for which there is a path of length at most $j$ in $G$.
\end{lemma}

\newcommand{\proofknowledgegraph}{
\begin{proof}\shortOnly{[Proof of \Cref{lem:knowledgegraph}]}
The first equation is simply stating that nodes do not know each other initially. To see the second equation we note that there are two ways nodes can learn about new nodes. The first one is by contacting a random node. This only adds edges in $G_{t+1}$ to the knowledge graph $K_t$. By initiating a contact along an edge in $G_{t+1}$ or along an edge in $K_t$ corresponding to a node known from previous rounds, nodes can also indirectly learn about nodes known to nodes they know. Even if all such ``neighbors'' are contacted and every node responds with everything they know (which would result in many message of very large size), a nodes does not learn about more nodes than its $2$-hop neighborhood. This is exactly what is stated in the second equation. The last equation now follows from the other two by induction and taking into account that for any two graphs $G,G'$ and any $j,j'>1$ we have $G^i \cup G' \subseteq (G \cup G')^i$ and of course $(G^j)^{j'} = G^{jj'}$.
\end{proof}
}\fullOnly{\proofknowledgegraph}

We can now prove \Cref{thm:LB}:

\begin{theorem}[\Cref{thm:LB} restated]
Suppose we allow unlimited messages sizes, non address-oblivious communication, and have no restriction on how many communications a node can be involved in or even to how many known nodes a node contact in one round. Still, any algorithm running in less than $T = \log \log n - \log \log \log n - \omega(1)$ rounds fails to spread a message that starts at one node $u$ to all nodes with high probability.
\end{theorem}
\begin{proof}
First we note that spreading a message starting at one node $u$ can be used to elect $u$ as a cluster leader by simply attaching its ID to the message spread. This could then be used to make $K_T$ the complete graph by informing all nodes about each other in at most two additional rounds, at least if the message size is unbounded. We show that this is, with high probability, impossible in less than $T$ rounds by using \Cref{lem:knowledgegraph}. In particular, it certainly holds that $K_T \subseteq {K'}^{{2}^{T}}$ where $K' = (\bigcup_{i=1}^T G_i)$. Note that $K'$ is simply a random graph that is created by each node sampling $T$ neighbors independently and adding an undirected edge to those neighbors. It is clear that $K_T$ is a complete graph if and only if the diameter of $K'$ is at most $2^T$. It is easy to see that for $T = \log \log n - \log \log \log n - \omega(1)$ the probability that such a random graph has diameter $o(\log n / \log \log n)$ is extremely small which completes the proof. An easy way to get such a high probability result is to say that the probability that there is a node in $K'$ with degree $\Delta = \omega(\log n)$ is at most $n^{-\omega(1)}$ and any graph with maximum degree $\Theta(\log n)$ necessarily has a diameter of $\Omega(\log n / \log \log n)$.
\end{proof}

\newcommand{\sectionBoundedCommunication}{
\section{Bounded Communications}\label{sec:boundedcommunication}

In this section we discuss and prove tradeoffs between the running time of any gossip algorithm and the maximum number of communications a node can be involved in per round.

In the standard gossip model which is also employed in this paper, each node can only initiate one PUSH or PULL communication per round. 
Still, a node might be involved in multiple communications initiated by other nodes. The possibility of this happening is present in the models of \emph{all} prior gossip works (discussed in \Cref{sec:intro}). In few cases, in particular the random phone call model with an underlying complete graph, having a large number of communications per round in one node is not crucial. For example, the running time of the uniformly random gossip on the complete graph remains $O(\log n)$ even if each node only accepts one communication per round (but the result of \cite{pittel1987spreading} which concentrates on the leading constants and lower order terms change). For most gossip works however such a concentration of communications is unavoidable. For example the works \cite{GiakkoupisS2012vertex,Giakkoupis2011,NC} which generalize the random phone call model to arbitrary topologies show that gossip is fast on any topology with good vertex or edge expansion. In particular, they imply that gossip completes in $\log^{O(1)} n$ rounds on a star graph because the star graph has a constant edge expansion $\Phi =O(1)$ and constant vertex expansion $\alpha =O(1)$. This is obviously not possible without the center node talking to at least $\Delta = n / \log^{O(1)} n$ nodes per round on average. Similarly, all results that use direct addressing crucially rely on each node being able to respond to many nodes and many of them can have one node communicate with (almost) all other nodes in the network in one round. This is also the case for \clusterone and \clustertwo. In fact our entire set of cluster-operations are based on having a cluster-center coordinating $\omega(1)$ cluster-followers in a single communication step.

While in some settings having $\Delta = n$ may be acceptable, in other settings one would like to have $\Delta$ to be not much more than a small polynomial or even (significantly) less. In what follows we explicitly address this parameter and design optimal algorithms for any $\Delta$.

We first remark that having a super-constant $\Delta$ is unavoidable if one wants to achieve a sub-logarithmic round complexity. In particular, it is easy to see that any global communication with the restriction that no node responds to more than $\Delta$ requests per round takes at least $\log n / \log \Delta$ rounds:

\begin{lemma}\label{lem:trivialLB}
Any algorithm in which no node participates in more than $\Delta$ communications per round requires at least $\log n / \log \Delta$ rounds to inform all nodes of a message.
\end{lemma}
\begin{proof}
Since each informed node can participate in at most $\Delta$ communications the number of informed nodes cannot grow by more than a $\Delta$ factor per round. Inductively the the number of informed nodes in round $r$ is therefore at most $\Delta^{r-1}$ and $\log n / \Delta$ rounds are needed to bring this quantity to $n$.
\end{proof}

In the remainder of this section we will show that the algorithmic ideas developed in this paper can be used to achieve any point on the this tradeoff curve between $\Delta$ and the required number of rounds while also maintaining the optimal message- and bit-complexity. We do this in two steps. We first define a $\Delta$-clustering and show that any such clustering easily allows $\Delta$-bounded communications in the desired $\log n / \log \Delta$ round complexity by running simple and robust gossip algorithms on top of this clustering. We then show that a $\Delta$-clustering can be computed in $\Theta(\log \log n)$ rounds:

\begin{definition}
A $\Delta$-clustering is a clustering in which every node is clustered and all clusters have a size between $\Delta$ and $2\Delta$.
\end{definition}

We will show in \Cref{lem:clusterPushPull} that such a $\Delta$-clustering indeed allows to efficiently spread information in the desired $O(\log n / \log \Delta)$ rounds using the ClusterPUSH-PULL protocol given in \Cref{alg:clusterPushPull}. For sake of simplicity we focus on the case $\Delta = \log^{\omega(1)} n$ since for smaller $\Delta$ the improvement of $\log \Delta$ is at most $O(\log \log n)$ and therefore essentially negligible. Similarly, in this section we assume for convenience that $\Delta < n^{0.9}$ as otherwise one can might as well set $\Delta = n$ and run \clustertwo. 

\begin{algorithm}[htb!]
\caption{ClusterPUSH-PULL($\Delta$)}
\begin{algorithmic}[1]

\State Compute $\Delta$-Clustering
\State ClusterShare(message)

\Statex

\For{$\Theta(\frac{\log n}{\log \Delta})$ iterations}			\label{delta:loop}		
	\State newly informed clusters: ClusterPUSH(message), ClusterShare(message) \label{delta:push}
\EndFor
	
\Statex

\State uninformed nodes: PULL message from a random node		\label{delta:pull}			\Comment{ClusterPULL}
\State ClusterShare(message)	\label{delta:share}

\end{algorithmic}
\label{alg:clusterPushPull}
\end{algorithm}

\begin{lemma}\label{lem:clusterPushPull}
For any $\Delta = \log ^{\omega(1)} n$ the ClusterPUSH-PULL algorithm given in \Cref{alg:clusterPushPull} spreads a message to all nodes, with high probability. It furthermore uses only $O(\frac{\log n}{\log \Delta})$ rounds and $O(n)$ messages once the $\Delta$-clustering is computed.
\end{lemma}
\begin{proof}[Proof Sketch. ]
We first address the round and message complexity and then prove correctness:

The claimed round complexity of $O(\frac{\log n}{\log \Delta})$ is easily verified since the only loop of the algorithm in Line~\ref{delta:loop} takes $3$ rounds per iteration and runs for $\Theta(\frac{\log n}{\log \Delta})$ iterations. All other instructions take only $O(1)$ rounds. It is also easy to see that each node processes at most $O(1)$ messages over the course of the algorithm.

Next we prove correctness, that is, that ClusterPUSH-PULL spreads any message to all nodes with high probability.
The first ClusterShare (Line~\ref{delta:push}) spreads the message to $\Delta = \omega(\log n)$ nodes. From there on, so long as at most $n/\Omega(\Delta)$ nodes are informed, each ClusterPUSH and ClusterShare in the main loop (of Line~\ref{delta:loop}) increases the number of  informed nodes with high probability by a $\Delta/O(1)$ factor. Hence, the set of informed nodes grows exponentially until it reaches $n/\Delta > (\Delta/O(1))^{\Theta(\frac{\log n}{\log \Delta})}$ nodes.

To see that this suffices for the ClusterPULL/ClusterShare in Lines \ref{delta:pull}--\ref{delta:share} to inform all remaining nodes with high probability we focus on one uninformed node. This node gets informed during the ClusterPULL if at least one of the nodes in its cluster PULLs the message from a node in Line~\ref{delta:pull}. Since there are at least $\Delta = \log  ^{\omega(1)} n$ such nodes of which each has an independent chance of $1/\Delta$ to PULL from an informed node this happens with high probability.

This shows that ClusterPUSH-PULL indeed informs all nodes about the message with high probability using only $O(\frac{\log n}{\log \Delta})$ rounds and $O(n)$ messages.
\end{proof}

What remains to show is that a $\Theta(\Delta)$-clustering can be computed efficiently. The process for generating the clustering is essentially a generalization of the \clustertwo\ algorithm, where we stop growing clusters at around size $\Delta$. Then for additional $O(\log\log n)$ rounds, we join all unclustered nodes to existing clusters. All the while, we never connect any leader to more than $\Delta$ followers, nor transmit more than $O(n)$ messages. The detailed pseudo-code description of the resulting algorithm \clusterthree($\Delta$) is given in \Cref{alg:cluster3}.

\begin{algorithm}[htb!]
\caption{\clusterthree($\Delta$) with $\Delta = \log^{\omega(1)} n$}
\begin{algorithmic}[1]

\State \Call{GrowInitialClusters}{}  \Comment{as in \Cref{alg:cluster2}}
\State \Call{SquareClusters}{to size $s \geq \sqrt{\Delta \log n}/C''$} \Comment{as in \Cref{alg:cluster2} until $s \geq \sqrt{\Delta \log n}/C''$}
\State \Call{MergeClusters}{} \Comment{merge to size $\Theta(\Delta/C'')$}
\State \Call{BoundedClusterPush}{} \Comment{similar to \Cref{alg:cluster1} but with continuous ClusterResize}
\State \Call{UnclusteredNodesPull}{} \Comment{as in \Cref{alg:cluster1}}
\State ClusterResize($\Delta / C''$)

\Statex

\Procedure{MergeClusters}{} 
\State ClusterActivate($10 \frac{s}{\Delta  / C''}$) 
\State active clusters: ClusterPUSH(clusterID)
\State inactive clusters: ClusterMerge(random received ID) 
\EndProcedure

\Statex
\Procedure{BoundedClusterPush}{} 
\State ClusterActivate(1)
\For {$\Theta(\log \log n)$ rounds}
                \State ClusterResize($\Delta / C''$)
                \State ClusterPUSH($follow$)
                \State unclustered nodes: set $follow$ to any received ID (if any)
                \State ClusterSize
                \If{cluster grows by less than a factor of $1.1$}
									\State deactivate cluster
                \EndIf
\EndFor 
\EndProcedure

\end{algorithmic}
\label{alg:cluster3}
\end{algorithm}

\begin{theorem}
For any given $\Delta = \log^{\omega(1)} n$ the algorithm \clusterthree($\Delta$), with high probability, computes a $\Theta(\Delta)$-clustering using $O(\log \log n)$-rounds and $O(n)$ messages while no node communicates with more than $\Delta$ nodes in any round.
\end{theorem}
\begin{proof}[Proof Sketch.]
The fact that Procedures {\sc GrowInitialClusters} and {\sc SquareClusters} result in $\Theta(n/\log n)$ clustered nodes grouped into 
$\Theta(s)$ sized clusters with $s = \Omega(\sqrt{\Delta \log n}/C'')$ follows the same way as for \clustertwo, and we do not repeat the arguments here.

Procedure {\sc MergeClusters} of \clusterthree is a bit different from \clustertwo, though the techniques are similar. It activates every cluster with probability $10 \frac{s}{\Delta  / C''}$ followed by a ClusterPUSH of the active clusters. Since $\Delta$ is small enough a Chernoff bound shows that roughly a $\frac{10C''}{\Delta}$ fraction of clusters is activated. Since $\Delta = \log^{\omega(1)} n$ each inactive cluster gets furthermore PUSHed by multiple active clusters to merge. Importantly, and differently to other merging phases, inactive clusters choose one to merge with among these clusters uniformly at random. A Chernoff type-bound shows that this way inactive clusters get assigned to active clusters evenly, i.e., each active cluster gets chosen by $\Theta(\frac{\Delta  / C''}{s})$ inactive clusters and will therefore grow to size about $\Delta / 10 C''$.

Subsequently, these clusters then recruit unclustered nodes in the {\sc BoundedClusterPush} call. After each recruiting step clusters measure whether they grew by at least a factor of $1.1$ and resize. As before this guarantees that all but a small constant fraction of the nodes joins a cluster while the number of messages is $O(n)$. Resizing active clusters guarantees that no cluster gets too big and no leader has to communicate too much.

As a last step, we run {\sc UnclusteredNodesPull} on the unclustered nodes and make them join the first cluster they PULL from. Up to constants every unclustered node joins a uniformly random cluster. Even without any size control every cluster grows by at most a small constant in this round. The final ClusterResize($\Delta/C''$) ensures that all clusters have the same size (up to a factor of two) as desired by a $\Theta(\Delta)$-clustering.

\end{proof}
}\fullOnly{\sectionBoundedCommunication}

\newcommand{\sectionFaultTolerance}{
\section{Fault Tolerance}\label{sec:faults}

In this section we discuss the robustness of our algorithms against random node failures.

As already discussed in the introduction, fault tolerance is one of the key advantages of the gossip paradigm. In contrast to many more structured information dissemination protocols, most gossip algorithms naturally continue working or degrade nicely under link and node failures. The way Karp et al. \cite{KarpSSV2000} captured at least some of this robustness formally was by considering oblivious node failures. They showed that $F$ failures, that are independent from the randomness used by the algorithm, lead to all but $o(F)$ live nodes being informed. Avin et al. gave a similar guarantee in \cite{AvinElDISC13}. We remark here that our gossip algorithms have the same failure tolerance:

\begin{theorem}
Algorithms \clusterone, \clustertwo, and \clusterthree maintain their guarantees in terms of running time, message- and bit-complexity and bounds on communication $\Delta$ even if up to $F = O(n)$ arbitrary but obliviously chosen nodes fail in the beginning. They produce a cluster (or $\Delta$-clustering) of all but $o(F)$ surviving nodes. All nodes in this cluster learn the message to be broadcast if any of these surviving nodes got the message initially.
\end{theorem}
\begin{proof}
Since all three algorithms are symmetric with respect to nodes one can think of the oblivious failures as random failures. The initialization phase of all three algorithm relies only on sampling which is not disturbed by random failures. During the cluster merging phase, there are anyway at least a constant fraction of unclustered nodes that do not take part in any actions and having $F$ extra nodes, namely the failed ones, in this group does not make a difference. Lastly, the pull phase analysis is carried in two parts. So long as at the beginning of a pull round the number $I$ of uninformed nodes is more than $F$, then with high probability there will remain $O(I^2/n)$ uninformed after the round. Hence, in this part, the uninformed fraction is squared in each round. Once all but $O(F)$ nodes are informed, at most a constant fraction of the pulls go to failed nodes which results in $\Theta\left(O(F)(F/n)^{\Theta(\log \log n)}\right) = O\left(\frac{F}{\log^{O(1)} n}\right) = o(F)$ non-failed nodes remaining uninformed after $\Theta(\log \log n)$ rounds, as claimed. Furthermore, the $O(n)$ message complexity is preserved during the pull phase. This follows from the fact that throughout the pull phase, each pulling node has an (independent) chance of at least $1/2$ to get informed. This leads to an expected constant number of messages per node with high probability to at most $O(n)$ messages over all nodes.
\end{proof}

}\fullOnly{\sectionFaultTolerance}

\section*{Acknowledgments} We thank Chen Avin and Robert Els\"asser for sharing an early version of their work and for helpful discussions.

\bibliographystyle{abbrv}
\bibliography{ref}

\shortOnly{

\appendix
\section*{APPENDICES}

\section{Proofs for \Cref{sec:simplealg}}\label{app:proofscluster1}

\proofinitialclusters

\proofclusteroneprog

\proofmergeallclusters

\proofpull

\proofclusterone

\section{Proofs of \Cref{sec:linear}}

\proofcltwosamp

\proofcltwoinit

\proofcltwomain

\proofboundedpush

\section{Proofs for \Cref{sec:LB}}

\proofknowledgegraph

\newpage

\section{Algorithm \clustertwo} \label{app:alg2}

\algorithmtwo

\newpage

\sectionBoundedCommunication

\sectionFaultTolerance

}

\end{document}